\documentclass[11pt,fleqn,leqno]{article}
\usepackage{amsmath,amssymb,amsthm}
\usepackage{palatino,pxfonts}
\usepackage{indentfirst,geometry,setspace,fancyhdr,sectsty}
\usepackage{natbib}
\usepackage{relsize,url}
\usepackage{graphicx}
\usepackage{accents}
\usepackage{amsfonts}
\usepackage[bottom]{footmisc}
\usepackage{indentfirst}
\usepackage{caption}
\newcounter{enumi2}

\newtheorem{example}[enumi2]{Example}

\newtheorem{prop}[enumi]{Proposition}

\begin{document}
\title{Ordinal Imitative Dynamics}
\author{George Loginov\footnote{Department of Economics, Augustana University, 2001 S Summit Ave, Sioux Falls, SD 57197, USA. Email: gloginov@augie.edu}}
\maketitle
\begin{singlespace}
\begin{abstract}
\noindent This paper introduces an evolutionary dynamics based on \textit{imitate the better realization} (IBR) rule. Under this rule, agents in a population game imitate the strategy of a randomly chosen opponent whenever the opponent`s realized payoff is higher than their own. Such behavior generates an ordinal mean dynamics which is polynomial in strategy utilization frequencies. We demonstrate that while the dynamics does not possess Nash stationarity or payoff monotonicity, under it pure strategies iteratively strictly dominated by pure strategies are eliminated and strict equilibria are locally stable. We investigate the relationship between the dynamics based on the IBR rule and the replicator dynamics. In trivial cases, the two dynamics are topologically equivalent. In Rock-Paper-Scissors games we conjecture that both dynamics exhibit the same types of behavior, but the partitions of the game set do not coincide. In other cases, the IBR dynamics exhibits behaviors that are impossible under the replicator dynamics. 
\end{abstract} 
\section{Introduction}
\noindent When information about the available strategies and their payoffs is limited, it may be reasonable for players in a population game to copy the behavior of their opponents if it yields or at least seems to yield better payoffs. Such copying gives rise to a family of imitative revision protocols, in which a player's decision to switch strategies depends on some summary of the relative performance of a random sample of opponents the player gets to observe. 

The two components that comprise any imitative protocol are the sampling procedure which determines the candidates to be imitated, and the conditional imitation rate which describes the likelihood of imitation given the information about the candidates' strategies and payoffs. With respect to the payoff information one can distinguish between protocols that rely on average payoffs and ones that only depend on realized payoffs from a small number of matches. 

This paper studies the imitative protocol with the fewest information requirements. A player who gets a revision opportunity observes one opponent from the population at random and switches to that opponent's strategy whenever the opponent's realized payoff is higher than his or her own. This revision rule labeled \textit{imitate the better realization} (IBR) was first studied in \cite{Izquierdo2013} in the context of two-strategy games. (We elaborate on their results in Section \ref{2str}.) It gives rise to ordinal mean dynamics since it ignores the magnitudes of payoff differences, and the resulting dynamics is polynomial in strategy utilization frequencies. 

The disregard of the payoff differences deprives the dynamics of some common cardinal properties. For instance, the Nash equilibria of the base game need not be the rest points of the dynamics, and the average payoffs need not improve along the solution trajectories. At a rest point, instead of equilibrating the average payoffs of the surviving strategies, the dynamics balances the flows to and from each surviving strategy. 

Despite not being monotone in average payoffs, the dynamics still eliminates pure strategies iteratively dominated by pure strategies. Weakly dominated strategies and pure strategies dominated by mixed strategies, on the contrary, may survive. In addition, we demonstrate that strict equilbria are locally stable.  

The dynamics generated by the IBR rule in many cases qualitatively resembles the replicator dynamics, which can be derived from the \textit{pairwise proportional imitation} rule of \cite{Sch98}. Under the PPI rule a revising agent observes one opponent from the population at random and switches to that opponent's strategy \textit{at a rate proportional to the payoff advantage of that strategy}. Thus, both the IBR and the PPI rule are based on comparisons of realized payoffs, but the magnitudes of payoffs matter only under the latter rule. 

In two-strategy games the IBR dynamics and the replicator dynamics are topologically equivalent: they have the same number of rest points, and their stability and convergence properties are the same. In the Rock-Paper-Scissors games both dynamics exhibit one of the three possible behaviors: global convergence to the rest point, global convergence to the boundary, or closed orbits around the rest point\footnote{In the first two cases we make a conjecture about global behavior based on local stability analysis and simulations. In the third case we prove the statement formally.}, but these behaviors need to be the same. In other cases, for instance in Zeeman's game, the number of interior rest points the two dynamics possess is different.  

The study of ordinal imitative dynamics originated with \cite{Hof95im} and \cite{Sch98}, in which the \textit{imitate if better} (IB) dynamics -- the average payoff counterpart of \textit{imitate the better realization} -- were introduced and their main properties established. In particular, \cite{Hof95im} demonstrates that in the Rock-Paper-Scissors games the IB dynamics behaves similarly to the replicator dynamics. A surprising result is that IB need not eliminate dominated strategies. 

Finally, put in the larger context, the \textit{imitate the better realization} rule can be viewed as an analogue of the word-of-mouth communication model (Ellison and Fudenberg (1993)) for strategic environments: agents can learn about the relative merits of strategies from others' experiences, but need not be able to find out the exact advantage of a particular strategy. For instance, one can learn about a better route from a neighbor or a better mode of behavior from an elder. 
\section{The \textit{Imitate the Better Realization} protocol and its properties}
\noindent Suppose that a continuum of agents of mass 1 are randomly matched to play a symmetric two-player game with the payoff matrix $A$. Let $S=\{1, \ldots, n\}$ be the set of strategies, and for $i, j \in S$ let $\pi_{ij}$ be the payoff of strategy $i$ against strategy $j$. At any instant, the population state $\mathbf{x} = (x_1, x_2, \ldots, x_n)$ describes the proportion of players choosing each strategy. The set of all such population states is the $(n-1)$-dimensional simplex $\Delta = \{ \mathbf{x} \text{ } | \sum_{i=1}^n x_i =1, x_i \ge 0\}$. 

Assume that the agents do not know the structure of the game, they are not aware of all the available strategies, and they don't keep the record of the strategies they used in the past or the payoffs they received in their previous interactions. In addition, they don't know the current population state and they are not capable of correctly anticipating the way it will evolve. The only piece of information they possess and are able to retain is their current payoff, and the only way they can learn about alternative modes of behavior is by observing the strategies of others. 

The only objective of the agents compatible with these assumptions is maximizing their current payoffs. As usual in a population setting, they are only able to switch their strategies infrequently, and once a strategy revision opportunity arises, the revising agent observes one opponent from the population at random and switches to that opponent`s strategy whenever the opponent`s realized payoff is higher than his or her own. 

Effectively, the revising agent gets to compare their current payoff to a sample payoff of some other strategy without learning the circumstance under which that sample payoff was obtained. Thus agents cannot distinguish between strategies that perform better on average and favorable circumstances in which worse strategies outperform better ones. Besides, upon switching to the candidate strategy the agent's payoff may differ from the sample payoff he or she got to observe. As a result, such ``blind'' imitation may not be improving in terms of average payoffs, and yet it eliminates dominated strategies.  

Given this \textit{imitate the better realization} revision protocol, the switch rate $\rho_{ij}$ from strategy $i$ to strategy $j$ can be expressed as the probability that a payoff drawn from the $j$-th row of the payoff matrix $A$ exceeds a payoff drawn from the $i$-th row, with the population state $\mathbf{x}$ serving as the probability distribution:
\[ 
\rho_{ij}(\mathbf{x}, A) = x_j \sum_{k=1}^n x_k \sum_{m=1}^n x_m \mathbf{1}_{\{\pi_{jk}>\pi_{im}\}} \tag{IBR}
\]
The population setting offers the following interpretation: the realized payoff of a revising agent playing strategy $i$ equals $\pi_{im}$ with probability $x_m$ for $m \in S$. Such an agent would observe a payoff realization $\pi_{jk}$ for strategy $j$ with probability $x_j x_k$ for $k \in S$. Summing over $m$ and $k$ and counting the cases in which the payoff to the candidate strategy $j$ is better yields the probability of switching $\rho_{ij}(\mathbf{x}, A)$.

A stochastic process that emerges when agents in the population independently receive revision opportunities can be approximated by its mean dynamics which describes the expected change in the proportion of agents playing each strategy (\cite{BenWei03}). The mean dynamics of the process governed by imitation of the better realization is 
\[
\dot{x}_i = \sum_{j=1}^n x_j \rho_{ji}(\mathbf{x}, A) - x_i \sum_{j=1}^n \rho_{ij}(\mathbf{x}, A)
\]
\[ = x_i \sum_{j=1}^n x_j \left(\sum_{k=1}^n x_k \sum_{m=1}^n x_m (\mathbf{1}_{\{\pi_{jk}<\pi_{im}\}}-\mathbf{1}_{\{\pi_{jk}>\pi_{im}\}})\right) \tag{IBRD}
\]

Thus in general the IBR dynamics is a quartic polynomial in $n$ variables $\{x_1, x_2, \ldots, x_n\}$, but in certain cases, as shown in the next subsection, it reduces to the replicator dynamics which is cubic in $x_i$ and is the mean dynamics for a number of more information-demanding imitative rules (Imitation via pairwise comparisons of \cite{Hel92} and \cite{Sch98}, imitation driven by dissatisfaction of \cite{BjoWei96}, and imitation of success of \cite{Hof95im}, see also Section 5.4.2 in \cite{Sand}).  
 
\subsection{Relation to the replicator dynamics}
The connection between the IBR dynamics and the replicator dynamics can be established via the \textit{pairwise proportional imitation} protocol (PPI) introduced in \cite{Sch98}. Under the PPI the revising agent observes one opponent from the population at random and imitates that opponent's strategy at a rate proportional to its payoff advantage. Compared to the PPI rule based on realized payoffs, the IBR rule suppresses any payoff differences, thus under the IBR rule all conditional switch rates to strategies that exhibit higher outcomes are the same, whereas under the PPI rule the switch rate is higher the higher the outcome. 

Under the PPI rule based on \textit{realized} payoffs, an agent whose strategy $i$ yields payoff $\pi_{im}$ for some $m \in S$ switches to strategy $j$ if the observed payoff realization $\pi_{jk}$ exceeds $\pi_{im}$, with a conditional switch rate proportional to the payoff advantage of $j$, which is $[\pi_{jk}-\pi_{im}]_+$. Accounting for the likelihood of each payoff realization $\pi_{jk}$ one can express the switch rate from strategy $i$ to strategy $j$ as 
\[ 
\rho_{ij}(\mathbf{x}, A) = x_j \sum_{k=1}^n x_k \sum_{m=1}^n x_m [\pi_{jk}-\pi_{im}]_+ \tag{$PPI_R$}
\]

The PPI rule based on \textit{average} payoffs generates the switch rates in the form  
\[ 
\rho_{ij}(\mathbf{x}, A) = x_j [\pi_j-\pi_i]_+ \tag{$PPI_A$}
\]
where $x_j$ is the probability of sampling an agent whose strategy is $j$, and the term $[\pi_j-\pi_i]_+$ is the (average) payoff advantage of strategy $j$ over $i$. Due to linearity of average payoffs in strategy utilization frequencies $x_i$, both kinds of proportional imitation generate the replicator dynamics as their mean dynamics (Theorem 3 of \cite{Sch98}):
\[
\dot{x}_i = x_i \left(\pi_i - \bar{\pi}\right) \tag{RD}
\]

The difference in the conditional switch rates under the IBR and the PPI rules leads to a significant dissimilarity in the resulting mean dynamics. The rest points of the replicator dynamics are the restricted equilibria of the underlying game, so the average payoffs to all active strategies are the same and the agents have no incentives to switch between them. The rest points of the IBR dynamics, on the other hand, are ``ordinal restricted equilibria'' in which the net flow for each strategy is zero. Yet it is possible that the flows between the active strategies are positive at a rest point, and the average payoffs to strategies need not be the same. Thus in general, the sets of the rest points for the replicator and the IBR dynamics do not coincide. However, as the following proposition states, when the cardinality of the set of payoffs is low the IBR dynamics coincides with the replicator dynamics up to a constant change of speed. 

\begin{prop}
\label{payoffs}
If the payoff matrix $A$ contains only two distinct payoffs, the IBR dynamics reduces to the replicator dynamics up to a constant change of speed.
\end{prop}
\begin{proof}
Suppose WLOG that the set of payoffs is $\{ \pi_{ij} \} = \{0, k\}$ for some $k>0$. Then for any pair of strategy profiles involving the strategies $i$ and $j$ the conditional switch rates from $i$ to $j$ under the IBR rule are proportional to those under the PPI rule. For any $i,j,k,m \in \{1,2, \ldots, n\}$ the following holds: 
\[
[\pi_{jk}-\pi_{im}]_+ = k \cdot \mathbf{1}_{\{\pi_{jk}>\pi_{im}\}}, 
\]
therefore $\rho_{ij}^{PPI}(\mathbf{x}, A) = k \cdot \rho_{ij}^{IBR}(\mathbf{x}, A)$, so the IBR dynamics is the replicator dynamics scaled by $\frac{1}{k}$. 
\end{proof}

The converse of Proposition \ref{payoffs} need not be true. Example \ref{coordination} contains a game with four distinct payoffs in which both protocols generate the same dynamics. 
\begin{example}
\label{coordination}
Consider the coordination game with the payoff matrix
\[\left(
\begin{tabular}{cc}
4 & 1 \\
3 & 2 
\end{tabular}\right)\]
Both the IBR and the PPI protocols generate the same dynamics: $\dot{x} = x(1-x) (2x-1)$. Under the PPI rule an agent with payoff $\pi_{22} =2$ who observes a candidate with payoff $\pi_{11}=4$ is twice as likely to switch to strategy 1 than an agent with payoff $\pi_{21} =3$, but at the same time an agent with payoff $\pi_{12} =1$ is twice as likely to switch to strategy 2 when he or she observes a candidate with payoff $\pi_{21}=3$ rather than a candidate with payoff $\pi_{22}=2$. These higher switch rates annihilate each other, so in the end the flows between the strategies are identical to those under the IBR rule, when all switch rates upon observing a better outcome are the same. $\square$
\end{example}

Section \ref{symrps} presents another example in which the two dynamics draw closer: in the standard Rock-Paper-Scissors game the IBR dynamics can be obtained from the replicator dynamics by a positive non-constant change of speed. But Example \ref{coordination} and the standard RPS game are an exception to the general rule. 

\subsection{Payoff monotonicity and payoff positivity}
In this section it is shown that the IBR dynamics need not preserve such cardinal properties as payoff monotonicity and payoff positivity. Payoff monotonicity (\cite{Nac90}) requires that the order of growth rates be the same as the order of average payoffs (if $\pi_i > \pi_j$ then $\frac{\dot{x_i}}{x_i} > \frac{\dot{x_j}}{x_j}$). Payoff positivity (\cite{Nac90}) is a weaker requirement that a strategy have a positive growth rate if and only if its payoff is higher than the average payoff in the population. Weak payoff positivity (\cite{Wei95}) is an even weaker requirement that among the strategies with above-average payoffs there is one with a positive growth rate. The next example demonstrates that all these properties are violated for the IBR dynamics even in two-strategy games:
\begin{example}
\label{paymon}
Consider the coordination game with the payoff matrix $A$
\[A=\left(
\begin{tabular}{cc}
10 & 0 \\
3 & 3 
\end{tabular}\right)\]
Let $x$ be the frequency of the first (top) strategy in the population. The mixed strategy equilibrium in game $A$ is $x^*=0.3$, and for all $x>x^*$ the average payoff of the first strategy is higher than that of the second strategy. But the IBR dynamics for $A$ is 
\[
\dot{x} = x(1-x)(2x-1),
\] 
so for all $x<0.5$ the proportion of agents playing the first strategy is decreasing. Thus, for instance, when $x=0.4$ $\pi_1=4 > 3= \pi_2$, but $\dot{x}(0.4) < 0$, and so the IBR dynamics is neither payoff monotone nor payoff positive. $\square$
\end{example}
Depending on the payoffs, any state $x^* \in (0,1)$ can be the mixed Nash equilibrium of the game which has the same order of payoffs as game $A$ from the Example \ref{paymon}. Yet for all such games the IBR dynamics selects $x=0.5$ as the rest point, so the monotonicity and positivity properties are violated precisely at the states between $x^*$ and $0.5$. 

When the population state is between $x^*$ and $0.5$, the first strategy already has a higher average payoff yet the majority of agents playing it receive the lowest payoff and thus would treat switching to the other strategy as an improvement. The switches in the opposite direction are less likely since the agents currently playing the second strategy would only imitate the minority of strategy 1 agents who currently receive the overall highest payoff. In the remainder of the state space the vector fields of the IBR dynamics and the replicator dynamics point in the same direction. For this to happen, a strategy with a higher average payoff needs to guarantee a better payoff for a larger share of agents than the other strategy. 

This intuition also paves the way for the next result: elimination of dominated strategies. If strategy 1 dominates strategy 2, then at any interior population state there will be a positive flow from 2 to 1, and the net inflow from any other strategy would be higher for 1 than for 2. Together these effects result in the ultimate extinction of strategy 2. 

\subsection{Elimination of dominated strategies}
This section sharpens the results on elimination of dominated strategies for imitative dynamics. As established by \cite{Nac90} and \cite{SamZha92}, ``cardinal'' imitative dynamics, including the replicator dynamics, eliminate pure strategies (iteratively) dominated by other pure strategies due to payoff monotonicity. The IBR dynamics, on the contrary, is a non-monotone imitative dynamics which still eliminates such dominated strategies. In terms of the comparison between the PPI and the IBR rules, this result means that imitation alone can be sufficient for the elimination of dominated strategies. In addition, \cite{HofSan11} demonstrate that under most dynamics not based on imitation dominated pure strategies can survive, in part because the agents in a population setting may be unable to recognize dominated strategies and thus avoid them. In the case of the IBR dynamics, the available payoff information is also insufficient to identify dominated strategies, and yet the agent switch away from them in the course of the play. 

\begin{prop}
\label{dom}
If a strategy is (iteratively) dominated by another pure strategy, then it is eliminated along any interior solution of the IBR dynamics.
\end{prop}
\begin{proof}
Suppose that strategy $j$ is dominated by strategy $i$. Fix $k \in \{1, \ldots, n\}$, by dominance $\pi_{ik}>\pi_{jk}$. Take a strategy $p \ne i, j$ and consider all possible strategy profiles $(p,q)$ that might arise in a match involving an agent playing this strategy. For a fixed $q \in \{1, \ldots, n\}$ the payoff $\pi_{pq}$ would fall into one of these three categories:
\begin{enumerate}
\item $\pi_{pq} \le \pi_{jk} < \pi_{ik}$, in which case there is an inflow into each strategy: $(x_p x_q x_k) x_i$ into strategy $i$ and $(x_p x_q x_k) x_j$ into $j$. Let $P_k x_i$ and $P_k x_j$ denote the total flows in this case. 
\item $\pi_{jk} < \pi_{pq} \le \pi_{ik}$, so there is outflow from strategy $j$ and inflow into strategy $i$, with total flows expressed as some $Q_k x_i$ and $Q_k x_j$.
\item $\pi_{jk} < \pi_{ik} < \pi_{pq}$, so there is outflow from both strategies, with total outflow expressed by $R_k x_i$ and $R_k x_j$.
\end{enumerate}
In addition, there is net inflow $Tx_i x_j$ from strategy $j$ to $i$ which includes at least the terms $x_k ^2 x_i x_j$.

In terms of these flow components the change in the population proportion for the strategies $i$ and $j$ can be expressed as 
\begin{align*}
\dot{x}_i &= \sum_{k=1}^n \left( P_k x_i + Q_k x_i - R_k x_i\right) + Tx_i x_j\\ 
\dot{x}_j &= \sum_{k=1}^n \left( P_k x_j - Q_k x_j - R_k x_j\right) - Tx_i x_j
\end{align*}
and so the change in the difference $d$ in growth rates of the two strategies is always positive:
\[
\dot{d}=\frac{\dot{x}_i}{x_i}-\frac{\dot{x}_j}{x_j} = 2 \sum_{k=1}^n Q_k + T(x_i+x_j) > 0
\]
Therefore, using the standard method (see, for instance, Proposition 1 in \cite{Vio15}), $d \to +\infty$ as $t \to \infty$, and thus strategy $j$ is eliminated.   

In the restricted game in which strategy $j$ is eliminated, one can apply the same reasoning to demonstrate that any strategy that becomes dominated will be eliminated as well. Thus by continuity of the IBR dynamics in the neighborhood of the edge opposite to the vertex $x_j=1$ in the original game (this edge corresponds to the simplex of the restricted game) the dynamics will select against the iteratively dominated strategies. See game $A_2$ in the example \ref{domin} for an illustration to this argument. 
\end{proof}

With a slight adjustment ($\pi_{ik}>\pi_{jk}$ would hold for at least one $k$, but not necessary all $k \in \{1, \ldots, n\}$) the argument can be applied to weakly dominated strategies as well, but one should only consider weakly dominated strategies after all strictly dominated strategies are eliminated. Otherwise, a strategy that is weakly dominated only with respect to a strictly dominated strategy may survive, as illustrated by game $A_1$ in the following example. 
\begin{example} \label{domin}
Consider\footnote{The figures in this and other examples are generated in \textit{EvoDyn-3s}. See \cite{SanIzIz18}.} the games $A_1$ and $A_2$
\begin{center}
\begin{tabular}{ccc}
$A_1=\left(
\begin{tabular}{ccc}
1 & 1 & 2 \\
1 & 1 & 1 \\
0 & 0 & 0
\end{tabular}\right)$ &&
$A_2=\left(
\begin{tabular}{ccc}
3 & 2 & 0 \\
2 & 1 & 3 \\
1 & 0 & 2
\end{tabular}\right)$
\end{tabular}
\end{center}

\begin{center}
\begin{figure}[h]
\centering
\begin{tabular}{ccc}
$\includegraphics[scale=0.5]{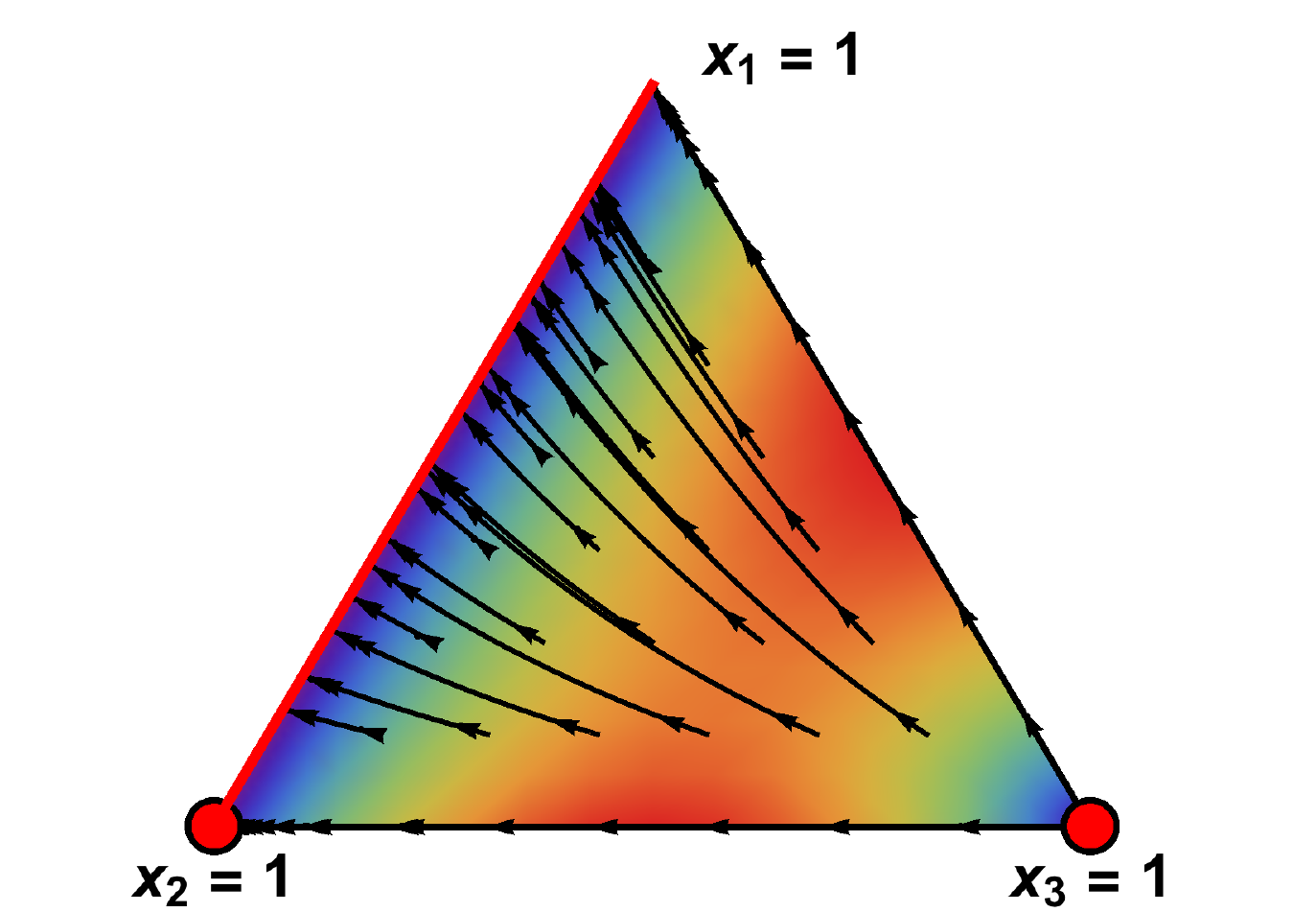}$ && $\includegraphics[scale=0.5]{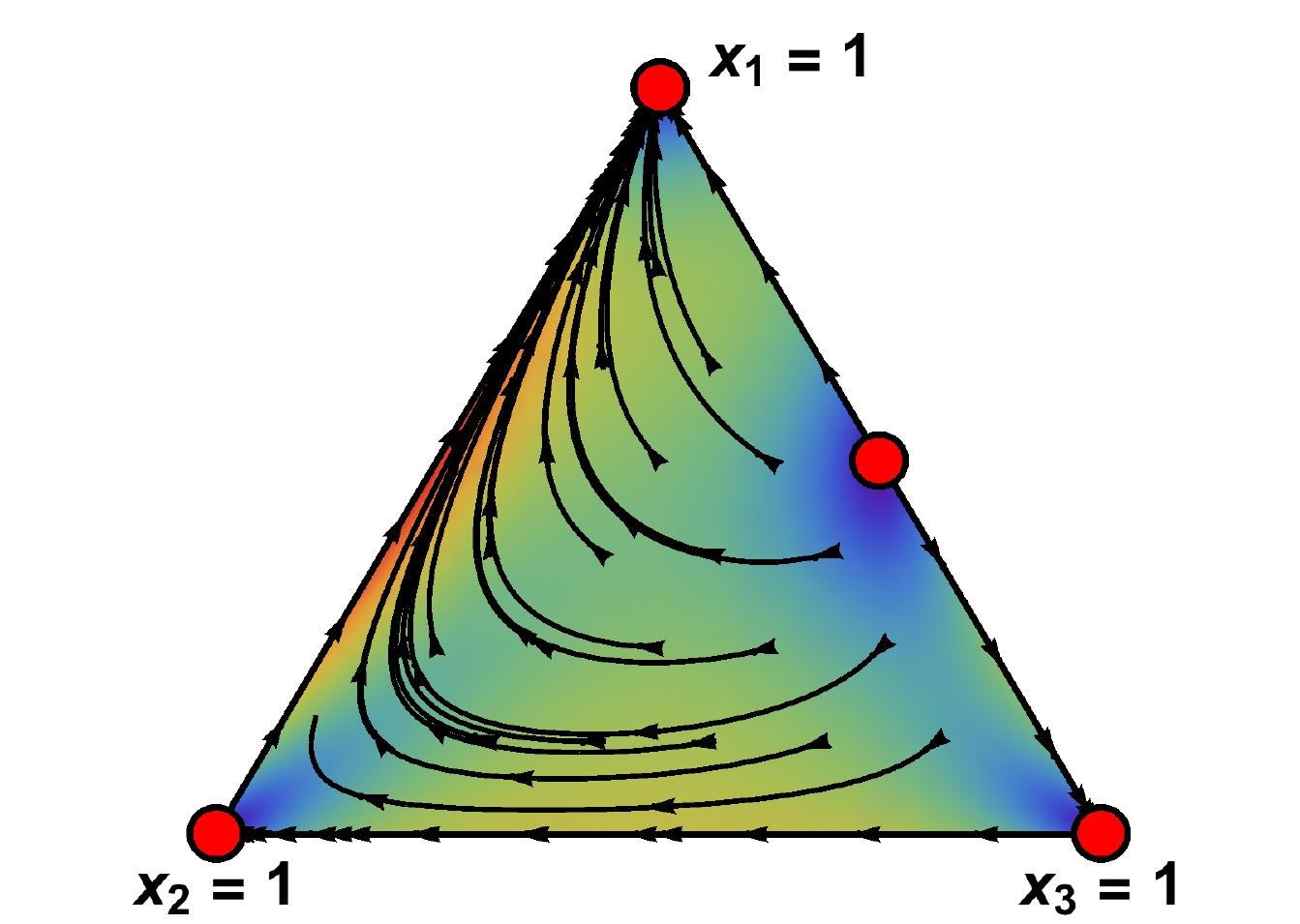}$ 
\end{tabular}
\caption{Some solution trajectories in games $A_1$ (left) and $A_2$ (right).}
\end{figure}
\end{center}

In the game $A_1$ strategy 3 is dominated by both 1 and 2, strategy 2 is weakly dominated by 1, but once strategy 3 is eliminated, both 1 and 2 coexist. 

From the perspective of an agent playing strategy 3 the remaining two strategies are equally good. The only case when strategy 1 gains advantage over strategy 2 is when an agent playing strategy 2 gets to imitate someone playing 1 against 3. The probability of observing such a candidate is $x_1x_3$, so the switch rate from strategy 2 to strategy 1 is low near the pure state $x_3=1$ and near the edge $x_3=0$. It is relatively high near the center of the simplex where the expression $x_1 x_2 x_3$ is maximized. 

In the game $A_2$ strategy 3 is dominated by strategy 2, and after strategy 3 is eliminated, strategy 2 is dominated by strategy 1. In this game the solution trajectories originating near the pure state $x_3=1$ first move in the direction of the pure state $x_2=1$, since when most agents choose strategy 3 almost no one gets to imitate strategy 1 as the majority of strategy 1 agents receive the lowest payoff. But after the population state gets sufficiently close to $x_2=1$ and the strategy 3 becomes almost extinct, agents begin to realize the advantage of 1 over 2. $\square$
\end{example}

Another example that complements the result of Proposition \ref{dom} demonstrates that a strategy dominated by mixed strategies may survive. 
\begin{example}
Consider the games $A_3$ and $A_4$ with $\alpha \in (1,4)$:  
\begin{center}
\begin{tabular}{ccc}
$A_3=\left(
\begin{tabular}{ccc}
4 & 4 & 1 \\
$\alpha$ & $\alpha$ & $\alpha$ \\
1 & 1 & 4
\end{tabular}\right)$ &&
$A_4=\left(
\begin{tabular}{ccc}
1 & 1 & 4 \\
$\alpha$ & $\alpha$ & $\alpha$ \\
4 & 4 & 1
\end{tabular}\right)$
\end{tabular}
\end{center}
Compared to game $A_3$, the order of payoffs in game $A_4$ is reversed. In both games strategy 2 is always the second best, and when $\alpha < 2.5$, it is dominated by a mixed strategy of 1 and 3.

The IBR dynamics in game $A_3$ is 
\begin{align*}
\dot{x} &= x(1-x)(1-2z)\\
\dot{y} &= y(z-x)(1-2z)\\
\dot{z} &= z(1-z)(2z-1)
\end{align*}
where $x, y$, and $z$ are the proportions of strategies 1, 2, and 3, respectively. When exactly half of the population $(z=\frac12)$ chooses strategy 3, strategy 2 can survive. In the game $A_3$, strategy 2 is otherwise eliminated: strategy 1 is better than 2 when $z<\frac12$, while 3 is better than 2 when $z>\frac12$. 
\begin{center}
\begin{figure}[h]
\centering
\begin{tabular}{ccc}
$\includegraphics[scale=0.5]{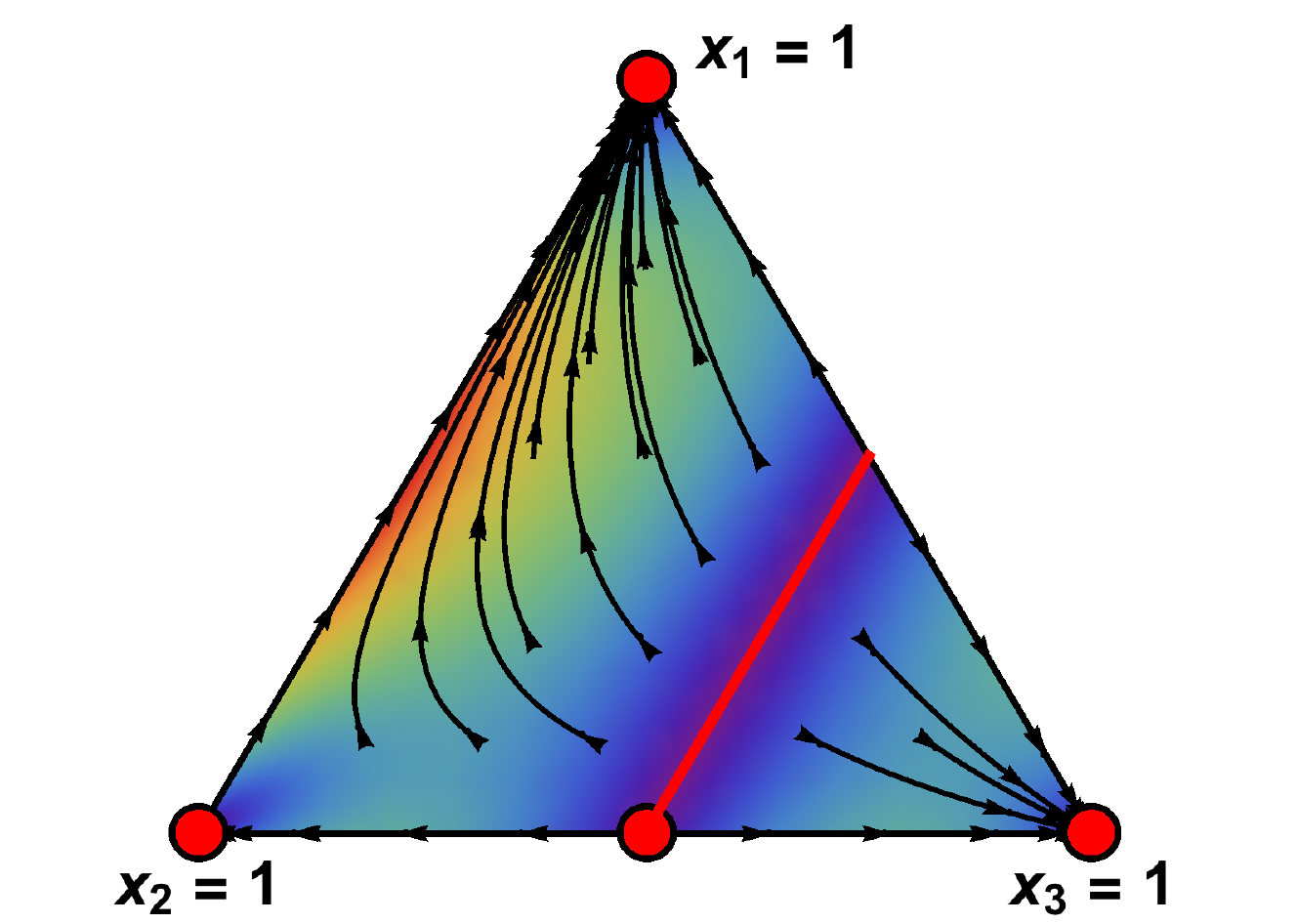}$ && $\includegraphics[scale=0.5]{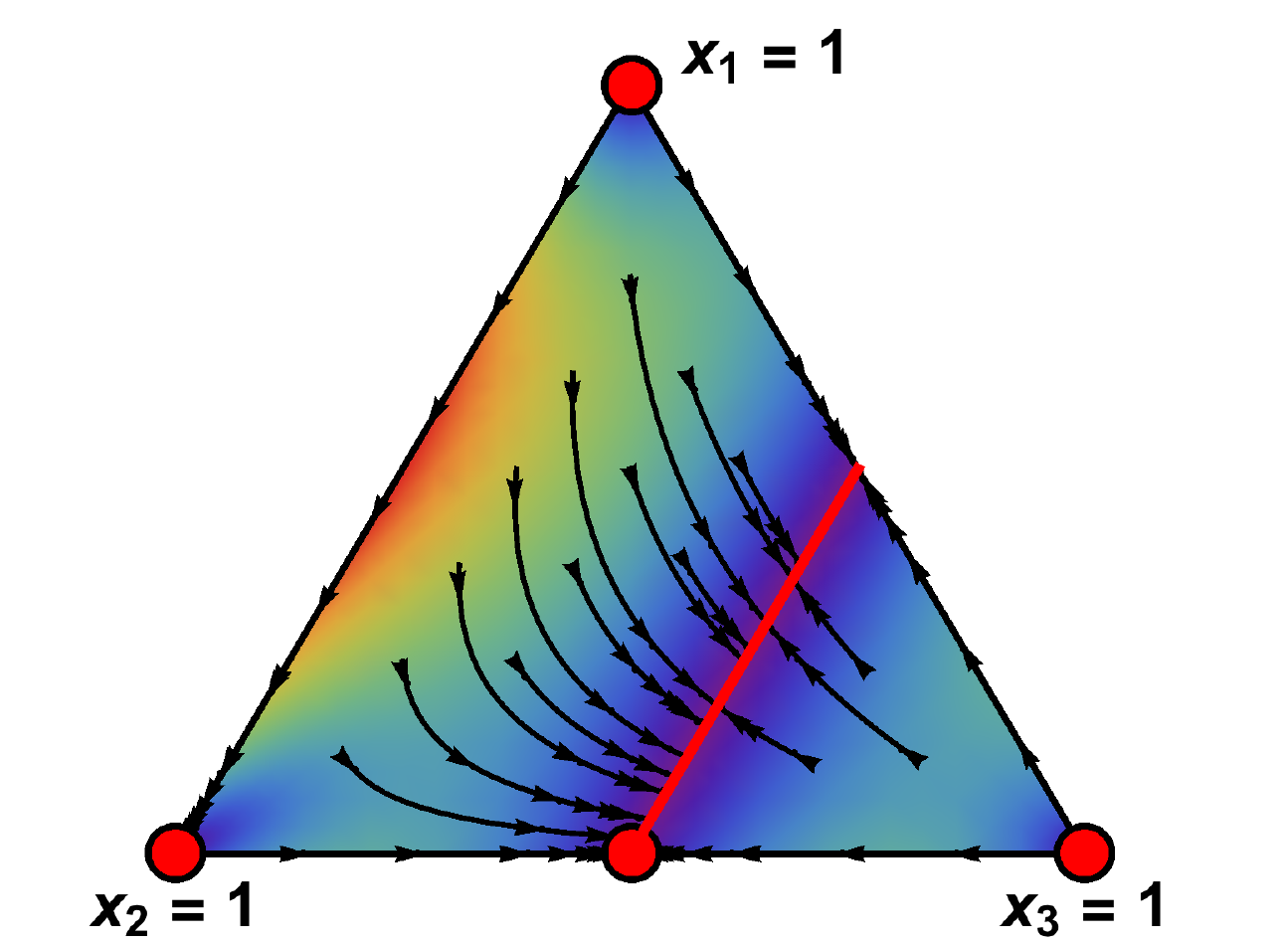}$ 
\end{tabular}
\caption{Some solution trajectories in games $A_3$ (left) and $A_4$ (right). The critical region is $z=\frac12$.}
\end{figure}
\end{center}
In game $A_4$ with the reversed order of payoffs the critical region $z=\frac12$ becomes absorbing, which suggests that strategy 2 survives along any trajectory originating in the interior of the state space.$\square$
\end{example}

\subsection{Stability of strict equilibria}
We conclude this section with a stability property for strict equilibria. If a strategy $i$ is the unique best response to itself, then in the neighborhood of a pure state $x_i=1$ any agent who currently employs a strategy $j \ne i$ would be most likely matched against an opponent playing strategy $i$, and upon receiving a revision opportunity would most likely observe a candidate earning $\pi_{ii}$ and as a consequence switch to $i$. The behavior of such agents would create an inflow into the strategy $i$ which is of higher order than any potential outflow caused by payoff advantages of other strategies over $i$, so the proportion of agents playing $i$ would increase, making the state $x_i=1$ a stable rest point. 

\setcounter{enumi}{2}
\begin{prop}
Strict symmetric equilibria are locally stable with a basin of attraction that includes all states with $x_i > 1- \tfrac{1}{\sqrt2}$.
\end{prop}
\begin{proof}
Suppose that the strategy profile $(i,i)$ is a strict equilibrium. To show that it is locally stable it is enough to demonstrate that $\dot{x}_i > 0$ in some neighborhood of the pure state $x_i=1$ as this implies that the function $L(x) = x_i$ is a strict local Lyapunov function for the state $x_i=1$. To do so, construct the lower bound on $\dot{x}_i$ by considering a game with the lowest net inflow into the strategy $i$. 

Let $x_i = 1- \epsilon$. Since $(i,i)$ is a strict equilibrium, for any $j \ne i$ we have $\pi_{ii} > \pi_{ji}$, so strategy $i$ would be imitated by any agent who currently obtains $\pi_{ji}$ and who observes a candidate obtaining $\pi_{ii}$. Such switches to strategy $i$ create an inflow of $x_i^3 \sum_{j \ne i} x_j = (1-\epsilon)^3 \epsilon$, which is a lower bound on the inflow into strategy $i$. 

To obtain a lower bound on $\dot{x}_i$, assume that in all other cases strategy $i$ performs worse than its alternatives, i.e. $\pi_{im} < \pi_{jk}$ for any $m \in S$ and $j,k \ne i$ and $\pi_{ik} < \pi_{ji}$ for any $j,k \ne i$. In the former case the outflow from strategy $i$ is $\sum_{m=1}^n x_i x_m \sum_{j,k \ne i} x_j x_k = (1-\epsilon) \epsilon^2$, and in the latter it is $\sum_{k \ne i} x_i x_k \sum_{j \ne i} x_j x_i = (1-\epsilon)^2 \epsilon^2$. The sum of these two components is the upper bound on the outflow from strategy $i$. 

Subtracting the highest outflow from the lowest inflow yields the desired lower bound:
\begin{align*}
\dot{x}_i &= \sum_{j=1}^n x_j \rho_{ji}(\mathbf{x}, A) - x_i \sum_{j=1}^n \rho_{ij}(\mathbf{x}, A) \\
	&= x_i \sum_{j=1}^n x_j \sum_{k=1}^n x_k \sum_{m=1}^n x_m \mathbf{1}_{\{\pi_{jk}<\pi_{im}\}} - x_i \sum_{j=1}^n x_j \sum_{k=1}^n x_k \sum_{m=1}^n x_m \mathbf{1}_{\{\pi_{jk}>\pi_{im}\}} \\
	& \ge x_i^3 \sum_{j \ne i} x_j \mathbf{1}_{\{\pi_{ji}<\pi_{ii}\}} - x_i \sum_{m=1}^n x_m \sum_{j,k \ne i} x_j x_k \mathbf{1}_{\{\pi_{jk}>\pi_{im}\}} - x_i^2 \sum_{k \ne i} x_k \sum_{j \ne i} x_j \mathbf{1}_{\{\pi_{ik}<\pi_{ji}\}}
\end{align*}
\begin{align*}
	& \ge (1-\epsilon)^3 \epsilon - (1-\epsilon) \epsilon^2-(1-\epsilon)^2 \epsilon^2\\
	& \ge \epsilon (1-\epsilon) \left(1-4\epsilon+2\epsilon^2\right) \\
	& \ge 0 \text{ when } \epsilon< 1- \tfrac{1}{\sqrt2}.
\end{align*}
Thus whenever $\epsilon< 1- \tfrac{1}{\sqrt2}$ the trajectory from an initial condition with $x_i=1-\epsilon$ converges to $x_i=1$, so any strict symmetric equilibrium is locally stable.  
\end{proof}
\section{Two-strategy games} \label{2str}
In this section we completely characterize the behavior of the IBR dynamics for two-strategy games. This topic was first studied in \cite{Izquierdo2013}, in which the dynamics is used to approximate the behavior of a finite population of agents who employ the IBR rule in the Hawk-Dove game. \cite{Izquierdo2013} also derive the general IBR equation for two-strategy games. Our paper complements their findings by identifying all possible rest points of the dynamics, and by demonstrating its equivalence to the replicator dynamics in terms of the number of rest points and the local behavior around them. 

For two-strategy games with two distinct payoffs the IBR dynamics and the replicator dynamics coincide. The table below presents all possible cases with 3 or 4 distinct payoffs, grouped by the game type: D -- game with a dominant strategy, W -- game with a weakly dominant strategy, C -- coordination game, A -- anticoordination game. 
\begin{center}
\begin{tabular}{cccccc}
$D_1$ & $D_2$ & $D_3$ & $D_4$ & $D_5$ & $D_6$ \\
$\left(
\begin{tabular}{cc}
4 & 3 \\
2 & 1 
\end{tabular}\right)$
&
$\left(
\begin{tabular}{cc}
4 & 3 \\
1 & 2 
\end{tabular}\right)$&
$\left(
\begin{tabular}{cc}
4 & 2 \\
3 & 1 
\end{tabular}\right)$
&
$\left(
\begin{tabular}{cc}
3 & 4 \\
2 & 1 
\end{tabular}\right)$
&
$\left(
\begin{tabular}{cc}
3 & 4 \\
1 & 2 
\end{tabular}\right)$
&
$\left(
\begin{tabular}{cc}
2 & 4 \\
1 & 3 
\end{tabular}\right)$\\
\\
$D_7$ & $D_8$ & $D_9$ & $D_{10}$ & $D_{11}$ & $D_{12}$ \\
$\left(
\begin{tabular}{cc}
3 & 3 \\
2 & 1 
\end{tabular}\right)$
&
$\left(
\begin{tabular}{cc}
3 & 3 \\
1 & 2 
\end{tabular}\right)$
&
$\left(
\begin{tabular}{cc}
3 & 2 \\
2 & 1 
\end{tabular}\right)$
&
$\left(
\begin{tabular}{cc}
2 & 3 \\
1 & 2 
\end{tabular}\right)$
&
$\left(
\begin{tabular}{cc}
3 & 2 \\
1 & 1 
\end{tabular}\right)$
&
$\left(
\begin{tabular}{cc}
2 & 3 \\
1 & 1 
\end{tabular}\right)$\\
\\
$W_1$ & $W_2$ & $W_3$ & $W_4$ & $W_5$ & $W_6$\\
$\left(
\begin{tabular}{cc}
3 & 2 \\
3 & 1 
\end{tabular}\right)$
&
$\left(
\begin{tabular}{cc}
2 & 3 \\
1 & 3 
\end{tabular}\right)$
&
$\left(
\begin{tabular}{cc}
3 & 2 \\
1 & 2 
\end{tabular}\right)$
&
$\left(
\begin{tabular}{cc}
2 & 3 \\
2 & 1 
\end{tabular}\right)$
&
$\left(
\begin{tabular}{cc}
3 & 1 \\
2 & 1 
\end{tabular}\right)$
&
$\left(
\begin{tabular}{cc}
1 & 3 \\
1 & 2 
\end{tabular}\right)$\\
\\
$C_1$ & $C_2$ & $C_3$ & $C_4$ & $C_5$ & $C_6$ \\
$\left(
\begin{tabular}{cc}
4 & 2 \\
1 & 3 
\end{tabular}\right)$
&
$\left(
\begin{tabular}{cc}
4 & 1 \\
3 & 2 
\end{tabular}\right)$
&
$\left(
\begin{tabular}{cc}
4 & 1 \\
2 & 3 
\end{tabular}\right)$
&
$\left(
\begin{tabular}{cc}
3 & 2 \\
1 & 3 
\end{tabular}\right)$
&
$\left(
\begin{tabular}{cc}
3 & 1 \\
2 & 2 
\end{tabular}\right)$
&
$\left(
\begin{tabular}{cc}
3 & 1 \\
1 & 2 
\end{tabular}\right)$\\
\\
$A_1$ & $A_2$ & $A_3$ & $A_4$ & $A_5$ & $A_6$\\
$\left(
\begin{tabular}{cc}
2 & 4 \\
3 & 1 
\end{tabular}\right)$
&
$\left(
\begin{tabular}{cc}
1 & 4 \\
3 & 2 
\end{tabular}\right)$
&
$\left(
\begin{tabular}{cc}
1 & 4 \\
2 & 3 
\end{tabular}\right)$
&
$\left(
\begin{tabular}{cc}
2 & 3 \\
3 & 1 
\end{tabular}\right)$
&
$\left(
\begin{tabular}{cc}
1 & 3 \\
2 & 2 
\end{tabular}\right)$
&
$\left(
\begin{tabular}{cc}
1 & 3 \\
2 & 1 
\end{tabular}\right)$
\end{tabular}
\captionof{table}{types of two-strategy games with 3 or 4 distinct payoffs}
\end{center}
In games with a (weakly) dominant strategy the set of rest points is $\{0,1\}$, and trajectories from any interior state converge to state $x=1$. 

\begin{center}
\begin{tabular}{|l|l|}
\hline
types & mean dynamics\\
\hline
$D_1, D_2, D_4, D_5, D_7, D_8, D_{11}, D_{12}$ & $\dot{x} = x(1-x)$\\
$D_3, D_6$ & $\dot{x} = x(1-x)[x^2+(1-x)^2]$\\
$D_9$ & $\dot{x} = x(1-x)[x+(1-x)^2]$\\
$D_{10}$ & $\dot{x} = x(1-x)[x^2+(1-x)]$\\
$W_1, W_6$ &  $\dot{x} = x(1-x)^3$\\
$W_2, W_5$ &  $\dot{x} = x^3(1-x)$\\
$W_3$ &  $\dot{x} = x^2(1-x)(1+x)$\\
$W_4$ &  $\dot{x} = x(1-x)^2(1+x)$\\
\hline
\end{tabular}
\captionof{table}{the mean dynamics for the games with (weakly) dominant strategies}
\end{center}
In coordination games the set of rest points is $\{0, 1-\frac{\sqrt{2}}{2}, \frac{3-\sqrt{5}}{2}, \frac12, 1\}$. The interior rest points are repelling, and trajectories from the interior states converge to the boundary states. 

\begin{center}
\begin{tabular}{|l|l|l|}
\hline
types & mean dynamics & interior RP\\
\hline
$C_1$ & $\dot{x} = x(1-x)[-2x^2+4x-1]$ & $x=1-\frac{\sqrt{2}}{2} \approx 0.293$\\
$C_2, C_3, C_5$ & $\dot{x} = x(1-x)[2x-1]$ & $x=\frac12$\\
$C_4, C_6$ & $\dot{x} = x(1-x)[-x^2+3x-1]$ & $x= \frac{3-\sqrt{5}}{2} \approx 0.382$\\
\hline
\end{tabular}
\captionof{table}{the mean dynamics for coordination games}
\end{center}
In anticoordination games the set of rest points is $\{0, \frac12, \frac{\sqrt{5}-1}{2}, \frac{1}{\sqrt{2}}, 1\}$. Essentially, the rest points in coordination and anticoordination games are the same save for the order of the strategies. The trajectories from the interior states converge to the interior rest point. 

\begin{center}
\begin{tabular}{|l|l|l|}
\hline
types & mean dynamics & interior RP\\
\hline
$A_1$ & $\dot{x} = x(1-x)[1-2x^2]$ & $x=\frac{1}{\sqrt{2}} \approx 0.707$\\
$A_2, A_3, A_5$ & $\dot{x} = x(1-x)[1-2x]$ & $x=\frac12$\\
$A_4, A_6$ & $\dot{x} = x(1-x)[-x^2-x+1]$ & $x= \frac{\sqrt{5}-1}{2} \approx 0.618$\\
\hline
\end{tabular}
\captionof{table}{the mean dynamics for anticoordination games}
\end{center}

Thus, the IBR dynamics exhibits the same properties as the replicator dynamics: in games with a dominant strategy both dynamics select it; in coordination games trajectories from the interior converge to one of the two pure states, while in anticoordination games trajectories converge to the interior rest point. 

As a corollary, this characterization also describes the behavior of the IBR dynamics on the boundary of the state space in a three-strategy game, since once one of the three strategies becomes extinct it is never reintroduced. Thus each boundary of the two-dimensional simplex can only have at most one rest point, unless the whole boundary is the rest area (which is possible in degenerate cases). Besides, this characterization shows that the interior rest points of the dynamics can only take one of the five possible values, and suggests that similar sets of rest points can be identified for games with more strategies.  
\section{Rock-Paper-Scissors games}
\subsection{Symmetric RPS games}  \label{symrps}
First consider the symmetric RPS game with the payoff matrix $A$ and $a,b>0$. 
\[A=\left(
\begin{tabular}{ccc}
0 & $-a$ & $b$ \\
$b$ & 0 & $-a$ \\
$-a$ &$b$ & 0
\end{tabular}\right)\]
For any values of the payoff parameters the game $A$ induces the same order of payoffs as the standard RPS game, for which the replicator dynamics is 
\begin{align*}
\dot{x} &= x(z-y),\\
\dot{y} &= y(x-z),\\
\dot{z} &= z(y-x),
\end{align*}
where $x$, $y$, and $z$ are the shares of agents playing Rock, Paper, and Scissors, respectively. 

The mean dynamics in game $A$ generated by the IBR protocol is 
\begin{align*} 
\dot{x} &= x(z-y)(1-xy-xz-yz),\\
\dot{y} &= y(x-z)(1-xy-xz-yz), \\ 
\dot{z} &= z(y-x)(1-xy-xz-yz)
\end{align*}

Thus in the standard RPS game the IBR dynamics is the replicator dynamics with speed adjusted by the positive non-constant function $(1-xy-xz-yz)$. This relationship helps identify the global behavior of the IBR dynamics in symmetric RPS games. 

\begin{prop}
In all symmetric RPS games the trajectories under the IBR dynamics are closed orbits around the unique interior rest point $(\frac13, \frac13, \frac13)$.
\end{prop}
\begin{proof}
Clearly, the only interior solution to the system (RPS) is $x^*=y^*=z^*=\frac13$. 

Since the IBR dynamics is the speed-adjusted replicator dynamics for this game, the Lyapunov function $H(\mathbf{x}) = x^* \log \frac{x}{x^*} + y^* \log \frac{y}{y^*} + z^* \log \frac{z}{z^*}$ (introduced in Theorem 6 in \cite{Zee80}) would be constant along the solutions of the IBR dynamics for all symmetric RPS games.
\end{proof}

Intuitively, in the standard RPS game the average payoff to a strategy (the information about the candidate strategy that a player receives under the proportional imitation rule based on the average payoffs) is equivalent to learning about the difference in shares of winners and losers under that strategy. So the average payoff under the PPI rule is higher whenever the likelihood of switching under the IBR rule is higher.    

\subsection{Ordered RPS games}
In general under the replicator dynamics the behavior of the system in the Rock-Paper-Scissors game solely depends on the determinant of the payoff matrix $A$. If $\det A = 0$, the solution trajectories form closed orbits around the interior steady state. If $\det A >0$, the interior steady state is a global attractor, whereas if $\det A < 0$ it is repelling (\cite{Zee80}).

Under the IBR dynamics we conjecture\footnote{We provide the proof of that statement for the closed orbits case, and state it as a conjecture for the remaining two cases based on simulations.} that the global behavior of the system can be one of the same three types: either all solutions converge to the interior steady state, or form closed orbits around it, or converge to the boundary. The difference is, the behavior depends on the order of payoffs, so for a fixed RPS game one can have any combination of behaviors under the replicator and the IBR dynamics. 

In Table 5 we consider the nine possible orderings over the payoffs in the RPS game. In all cases \textit{Rock} yields the highest positive payoff. 
\begin{center}
\begin{tabular}{ccc}
$A_1=\left(
\begin{tabular}{ccc}
0 & -3 & 3 \\
2 & 0 & -2 \\
-1 & 1 & 0
\end{tabular}\right)$
&
$A_2=\left(
\begin{tabular}{ccc}
0 & -2 & 3 \\
2 & 0 & -1 \\
-3 & 1 & 0
\end{tabular}\right)$
&
$A_3=\left(
\begin{tabular}{ccc}
0 & -1 & 3 \\
2 & 0 & -3 \\
-2 & 1 & 0
\end{tabular}\right)$\\
\\
$B_1=\left(
\begin{tabular}{ccc}
0 & -3 & 3 \\
1 & 0 & -1 \\
-2 & 2 & 0
\end{tabular}\right)$
&
$B_2=\left(
\begin{tabular}{ccc}
0 & -1 & 3 \\
1 & 0 & -2 \\
-3 & 2 & 0
\end{tabular}\right)$
&
$B_3=\left(
\begin{tabular}{ccc}
0 & -2 & 3 \\
1 & 0 & -3 \\
-1 & 2 & 0
\end{tabular}\right)$\\
\\
$C_1=\left(
\begin{tabular}{ccc}
0 & -1 & 3 \\
2 & 0 & -2 \\
-3 & 1 & 0
\end{tabular}\right)$
&
$C_2=\left(
\begin{tabular}{ccc}
0 & -3 & 3 \\
2 & 0 & -1 \\
-2 & 1 & 0
\end{tabular}\right)$
&
$C_3=\left(
\begin{tabular}{ccc}
0 & -2 & 3 \\
2 & 0 & -3 \\
-1 & 1 & 0
\end{tabular}\right)$\\
\end{tabular}
\captionof{table}{The interior steady state is repelling in A1-A3, an attractor in B1-B3, and a center in C1-C3.}
\end{center}
Each B game is obtained from an A game with the same index by reversing the order of payoffs and subsequently relabelling the strategies. This procedure reverses the flows along the solution trajectories, so the repelling rest points in games of type A become attractors in games of type B. The next proposition states this result formally. 
\begin{center}
\begin{figure}[h]
\centering
\begin{tabular}{ccc}
$\includegraphics[scale=0.42]{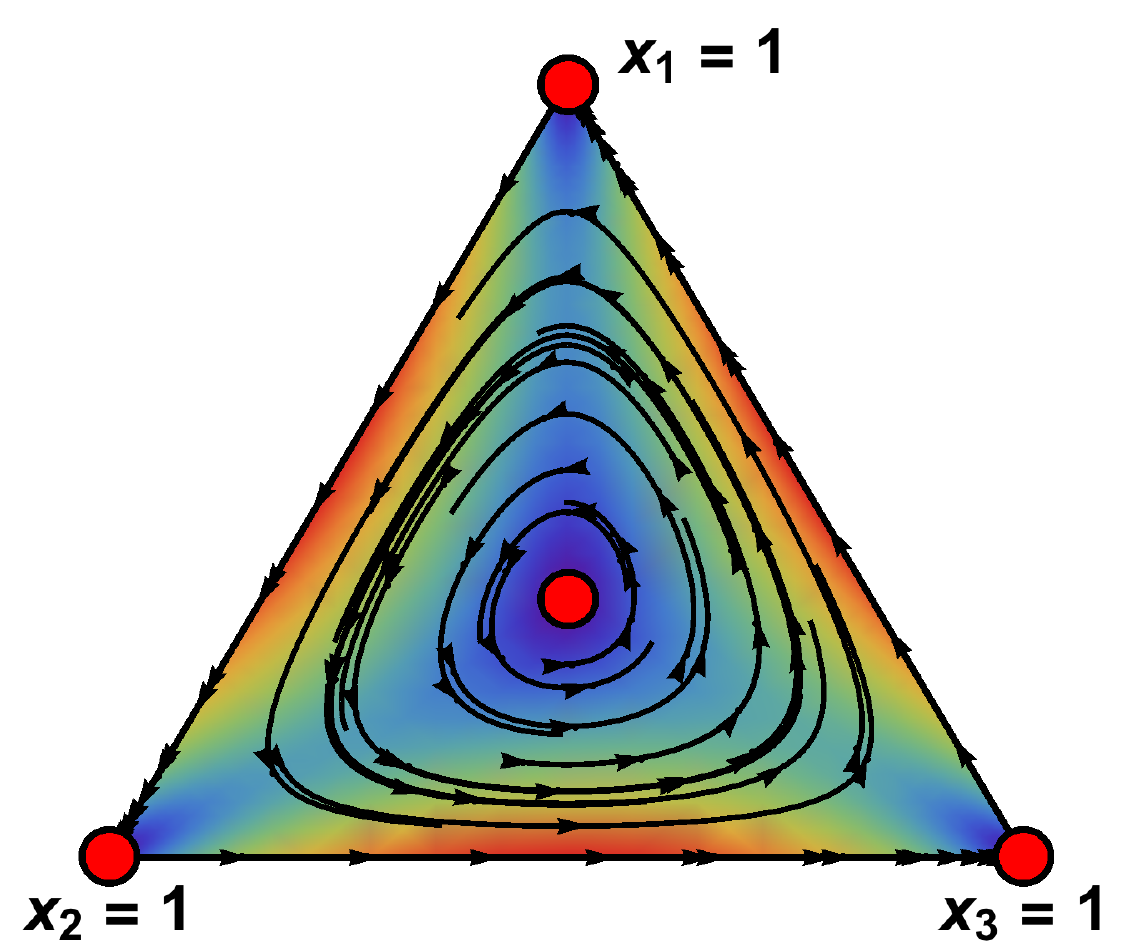}$ & $\includegraphics[scale=0.42]{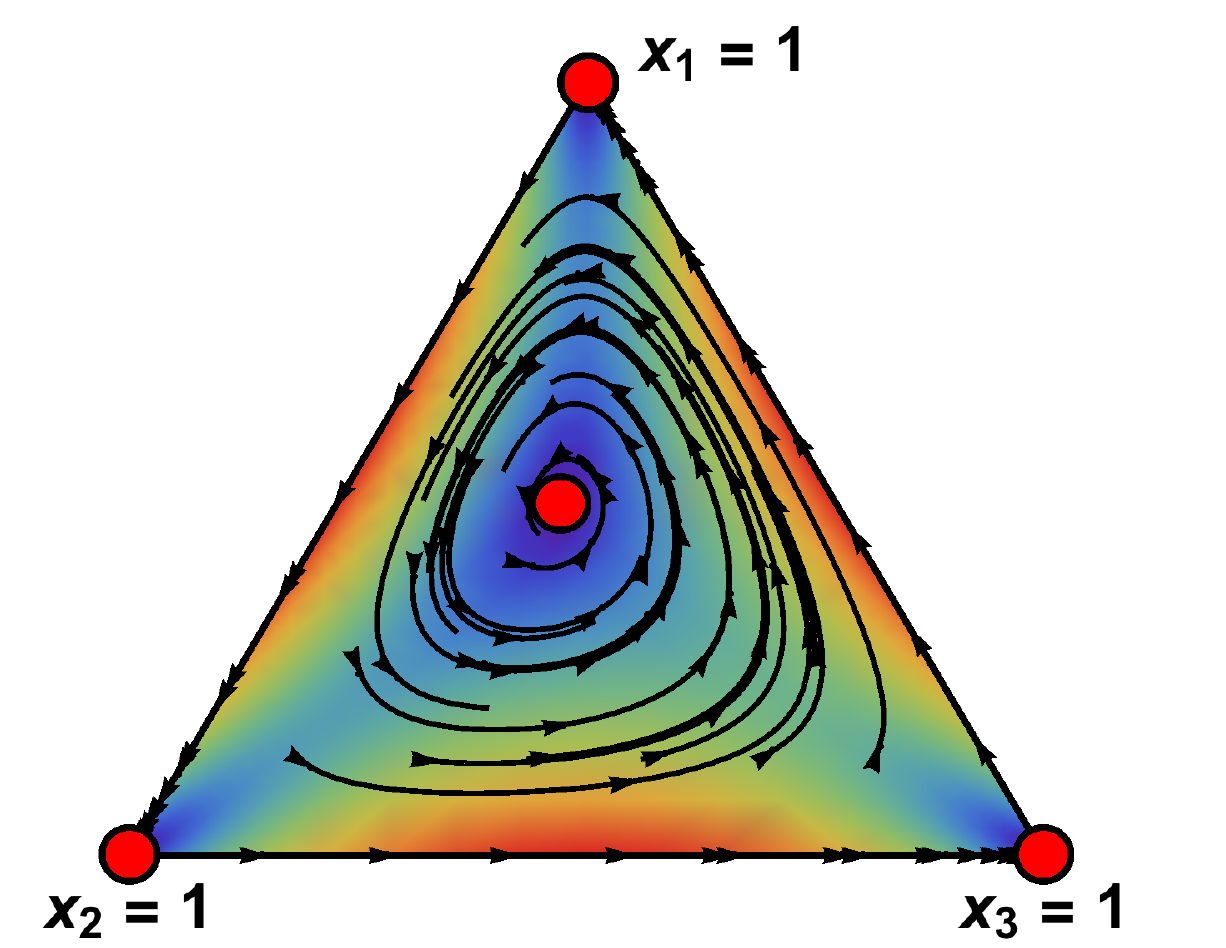}$ & $\includegraphics[scale=0.42]{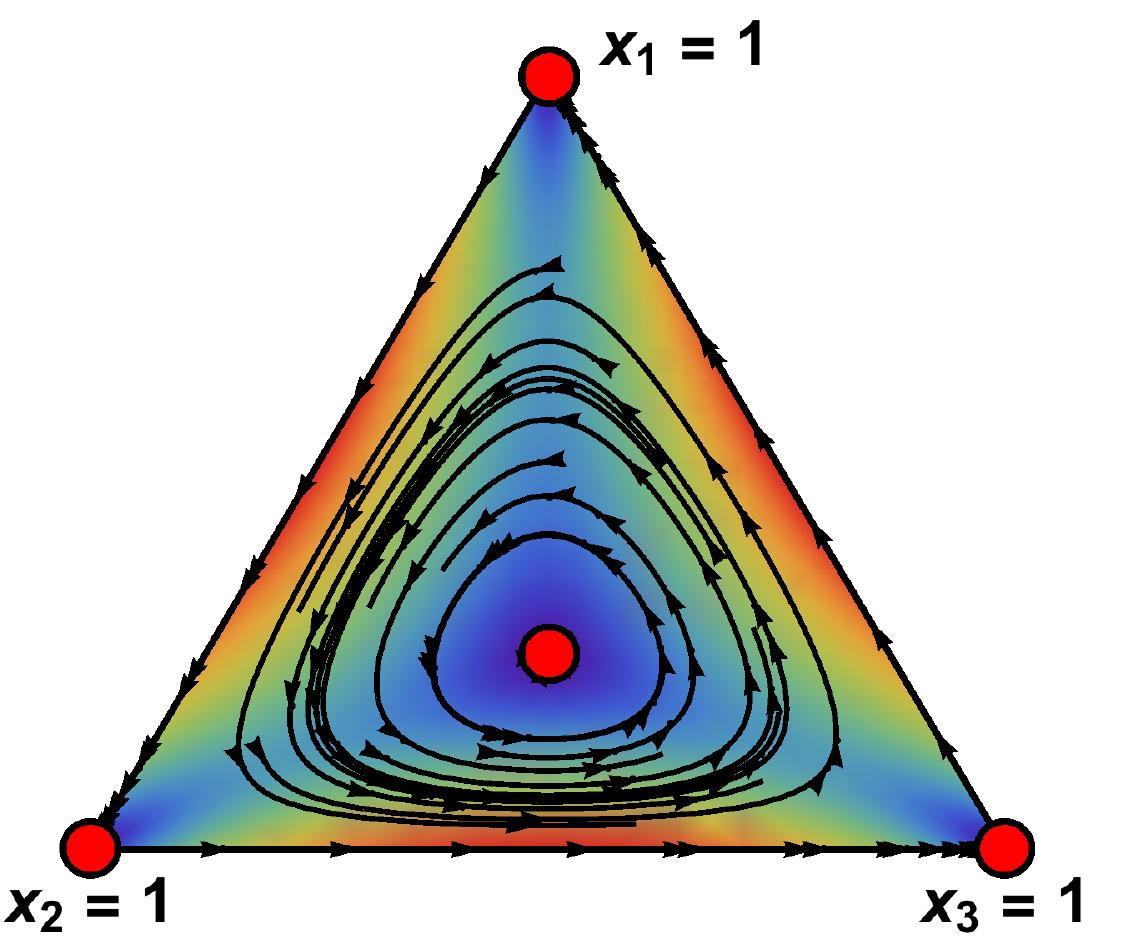}$
\end{tabular}
\caption{Some solution trajectories for games A1, B3, and C2.}
\end{figure}
\end{center}

\begin{prop}
i) The unique interior rest point in any A game is repelling.\\
ii) For any $i \in \{1,2,3\}$ the game $B_i$ can be obtained from the game $(-A_i)$ by relabeling the strategies.  
\end{prop}
\begin{proof}
i) The statement is proved for the game $A_1$ by direct computation. The proofs for games $A_2$ and $A_3$ are similar. 

The IBR dynamics in game $A_1$ can be written as 
\begin{align*} 
F(x,y,z) = \left(
\begin{matrix}
\dot{x}\\
\dot{y}\\
\dot{z}
\end{matrix}\right)
= \left(
\begin{matrix}
x(z-y)(x^2+y+z)\\ 
y(x-z)(x^2+y+z)+yz(y-x)(x+z)\\
z(y-x)(x^2+y+z)-yz(y-x)(x+z)
\end{matrix}\right)
\end{align*}
To identify the interior rest points notice that $\dot{x}=0$ requires $y=z$, in which case $\dot{y}=0$ reduces to
\begin{align*} 
(x-y) \left(x^2+2y -y(x+y)\right)=0.
\end{align*}
Plugging in $x=1-2y$ results in the equation
\begin{align*} 
(1-3y) \left(5y^2-3y+1\right)=0,
\end{align*}
with the only real solution $y=\frac13$. Thus the unique interior rest point of the system $F$ is $(\frac13, \frac13, \frac13)$. 

To identify the local behavior of the system $F$ around the interior steady state, project it from $\mathbb{R}^3$ onto $\Delta_3$, the two-dimensional simplex, to obtain the system
\begin{align*} 
\hat{F}(x,y) = \left(
\begin{matrix}
x(1-x-2y)(x^2-x+1)\\ 
y(2x+y-1)(y^2-y+1)+2xy(y-x)(1-x-y) 
\end{matrix}\right)
\end{align*}
The Jacobian of this system evaluated at the rest point $(\frac13, \frac13)$ is 
\begin{align*} 
D\hat{F}\left(\frac13, \frac13\right)=\frac{1}{27} \left(
\begin{matrix}
-7 &-14 \\
12 & 9
\end{matrix}\right)
\end{align*}
with the eigenvalues $\frac{1\pm2i\sqrt{26}}{27}$. Since both eigenvalues have positive real parts the rest point is repelling. \\
ii) To see the relationship between $A_1$ and $B_1$, write the matrices $A_1, -A_1$, and $B_1$ side by side: 
\begin{center}
\begin{tabular}{ccc}
$A_1=\left(
\begin{tabular}{ccc}
0 & -3 & 3 \\
2 & 0 & -2 \\
-1 & 1 & 0
\end{tabular}\right)$
&
$-A_1=\left(
\begin{tabular}{ccc}
0 & 3 & -3 \\
-2 & 0 & 2 \\
1 & -1 & 0
\end{tabular}\right)$
&
$B_1=\left(
\begin{tabular}{ccc}
0 & -3 & 3 \\
1 & 0 & -1 \\
-2 & 2 & 0
\end{tabular}\right)$
\end{tabular}
\end{center}
The game $B_1$ can be obtained from $(-A_1)$ by relabeling strategies 2 and 3. Formally the IBR dynamics in game $B_1$ can be written as 
\begin{align*} 
G(x,y,z) = \left(
\begin{matrix}
\dot{x}\\
\dot{y}\\
\dot{z}
\end{matrix}\right)
= \left(
\begin{matrix}
x(z-y)(x^2+y+z)\\ 
y(x-z)(x^2+y+z)-yz(x-z)(x+y)\\
z(y-x)(x^2+y+z)+yz(x-z)(x+y)
\end{matrix}\right)
\end{align*}
so $-F(x,y,z) = G(x, z, y)$. Thus the system $G$ is the time-reversed system $F$, so the eigenvalues of the Jacobian evaluated at the interior rest point of $\hat{G}$ both have negative real parts, and that rest point is an attractor. 
\end{proof}

Games of type $C$ require a different approach, since in them the eigenvalues of the Jacobian at the interior rest point are purely imaginary, so the local stability analysis using the Jacobian does not produce an unambiguous result. However, this obstacle can be overcome once one notices that up to the strategy labels, reversing the order in any C game results in the same game. 
\begin{prop} \label{cgame}
The unique interior rest point in any C game is a center, and any trajectory from the interior forms a closed orbit around it. 
\end{prop}
\begin{proof}
The proposition is proved for the game $C_2$. The proof for games $C_1$ and $C_3$ is similar.

Observe that the negative of the game $C_2$ is $C_2$ with strategies 2 and 3 interchanged. 
\begin{center}
\begin{tabular}{ccc}
$C_2=\left(
\begin{tabular}{ccc}
0 & -3 & 3 \\
2 & 0 & -1 \\
-2 & 1 & 0
\end{tabular}\right)$
&
&
$-C_2=\left(
\begin{tabular}{ccc}
0 & 3 & -3 \\
-2 & 0 & 1 \\
2 & -1 & 0
\end{tabular}\right)$
\end{tabular}
\end{center}

Formally, the IBR dynamics in $C_2$ 
\begin{align*} 
H(x,y,z) = \left(
\begin{matrix}
\dot{x}\\
\dot{y}\\
\dot{z}
\end{matrix}\right)
= \left(
\begin{matrix}
x(z-y)(x^2+y+z)\\ 
y(x-z)(x^2+y+z)-yz(x^2-xy-xz-yz)\\
z(y-x)(x^2+y+z)+yz(x^2-xy-xz-yz)
\end{matrix}\right)
\end{align*}
has the property $-H(x,y,z) = H(x, z, y)$. 

To identify the interior rest points of the system $H$ observe that $\dot{x}=0$ requires $z=y$, so that $x=1-y-z=1-2y$. Plugging the expressions for $x$ and $z$ into $\dot{y}=0$ yields the equation
\[
1-6y+16y^2-19y^3=0,
\]
which only has one root $y \approx 0.374$ in the interval $[0,1]$. Thus $\mathbf{x^*} \approx (0.252, 0.374, 0.374)$ is the unique interior rest point of the system $H$.

To show that the solution trajectories originating in the interior of the simplex form closed orbits around the rest point $\mathbf{x^*}$, we first show that any such solution circles around $\mathbf{x^*}$ and then apply the ``self-negating'' property to conclude that any circular solution trajectory must be indeed a closed orbit. 

\begin{figure}[h]
\begin{center}
\includegraphics{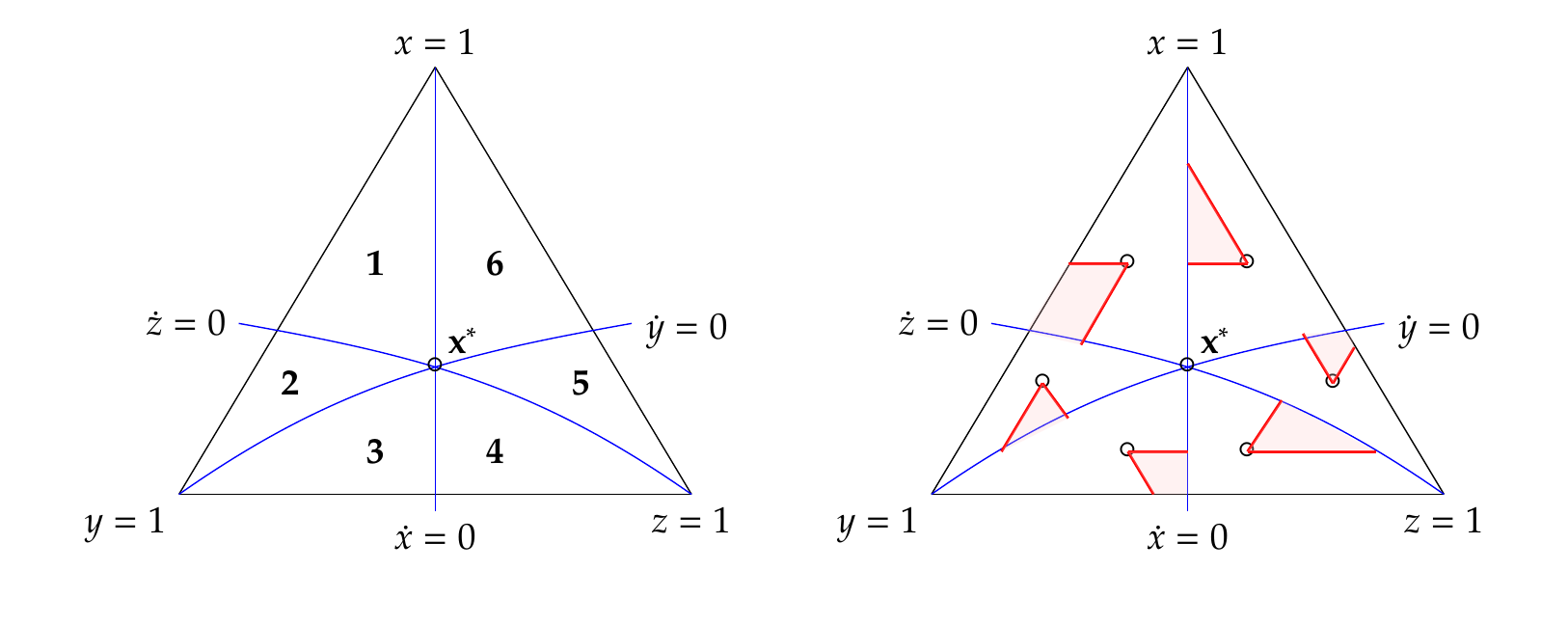}
\end{center}
\caption{Left: the nullclines divide the simplex into 6 regions. Right: solution trajectories from the marked points have to remain within the red areas.}
\end{figure}
Observe that the x-nullcline is the x-bisector, while the y-nullcline and z-nullcline originate at states $y=1$ and $z=1$, respectively, pass through the interior rest point $\mathbf{x^*}$ and hit the opposite edges of the simplex. Together the nullclines divide the simplex into six regions: $x$ is decreasing in regions 1, 2, and 3, and increasing in regions 4, 5, and 6; $y$ is decreasing in 3, 4, and 5, and increasing in 6, 1, and 2; $z$ is decreasing in regions 5, 6, and 1, and increasing in regions 2, 3, and 4.

Given the signs of the nullclines, the possible directions of motion in each of the six regions are restricted to a particular 60-degree wedge. For instance, in region 1, $x$ and $z$ must be decreasing, while $y$ is increasing, so a solution originating from a point inside region 1 can only move toward the boundary $z=0$ or the nullcline $\dot{z}=0$. But as the solution gets closer to the boundary, the speed $\dot{z}$ approaches 0, while the speeds $\dot{x}$ and $\dot{y}$ are bounded away from 0 within the red area that restricts the feasible directions of motion (Formally, as $z \rightarrow 0$, $\dot{x} \rightarrow -xy(x^2+y)$; $\dot{y} \rightarrow xy(x^2+y)$; $\dot{z} \rightarrow 0$, so near the edge $z=0$ the direction of motion is almost parallel to the edge.) Thus a solution originating inside region 1 has to exit it via the z-nullcline. Similarly, the solutions originating in regions 3 and 5 have to escape them via the x- and the y-nullclines, correspondingly. 

Solutions originating in region 2 cannot reach the boundary as $z$ must be increasing. But it is not immediately obvious that they cannot hit the rest point $\mathbf{x^*}$. To exclude this possibility, one has to show that the y-component of every point in the region 2 is at least as high as the y-component of $\mathbf{x^*}$, so that $\mathbf{x^*}$ can not be reached as $y$ must be increasing. This will be the case if the slope of the z-nullcline at $\mathbf{x^*}$ is not lower than the slope of the line $y=const$. To compute that slope, simplify the system $H$ to  
\begin{align*} 
H(x,y,z) = \left(
\begin{matrix}
\dot{x}\\
\dot{y}\\
\dot{z}
\end{matrix}\right)
= \left(
\begin{matrix}
x(z-y)(x^2+y+z)\\ 
y(x-z)(x^2+y+z^2)+2xy^2z\\
z(y-x)(x^2+y^2+z)-2xyz^2
\end{matrix}\right)
\end{align*}
so that the equation for the z-nullcline becomes $(y-x)(x^2+y^2+z)-2xyz=0$. Using $z=1-x-y$, apply the Implicit Function Theorem to compute the slope of the z-nullcline:
\begin{align*} 
\frac{dy}{dx} = \frac{x^2+y^2+z-(y-x)(2x-1)+2yz-2xy}{x^2+y^2+z+(y-x)(2y-1)-2xz+2xy}
\end{align*}
The rest point $\mathbf{x^*}$ is characterized by $z=y$ and $x=1-2y$, so the slope of the z-nullcline at $\mathbf{x^*}$ can be expressed solely in terms of $y$: 
\begin{align*}
\frac{dy}{dx}|_{\mathbf{x}=\mathbf{x^*}} &= \frac{(1-2y)^2+y^2+y-(3y-1)(1-4y)+2y^2-2y(1-2y)}{(1-2y)^2+y^2+y+(3y-1)(2y-1)}\\
			&= \frac{23y^2-12y+2}{11y^2-8y+2} >0 
\end{align*}
Thus at the rest point $\mathbf{x^*}$ the z-nullcline has a positive slope, whereas the slope of the line $y=const$ is 0 (in standard coordinates). Therefore at any point in the interior of region 2, $y > y(\mathbf{x^*})$, and the trajectories originating in that region have to escape it via the y-nullcline.

Similarly, in region 4 the slope of the x-nullcline $z-y=0$ equal to $-\tfrac12$ exceeds the slope of the line $z=const$ equal to $-1$, and in the region 6 the slope of the y-nullcline $(x-z)(x^2+y+z^2)+2xyz=0$ is 
\begin{align*} 
\frac{dy}{dx} = -2 \frac{x^2+y+z^2+(x-z)^2+2yz-2xy}{x^2+y+z^2+(x-z)(1-2z)+2xz-2xy}
\end{align*}
so that at $\mathbf{x^*}$ it becomes
\begin{align*}
\frac{dy}{dx}|_{\mathbf{x}=\mathbf{x^*}} &= -2 \frac{17y^2-10y+2}{11y^2-8y+2} < 0
\end{align*}
whereas the line $x=const$ is vertical. Hence the trajectories originating in both regions 4 and 6 have to escape them via the nullclines. Therefore the solution trajectory from any interior initial condition circles around the interior rest point $\mathbf{x^*}$ by sequentially entering and exiting each of the six regions via the nullclines. 

\begin{figure}[h]
\begin{center}
\includegraphics{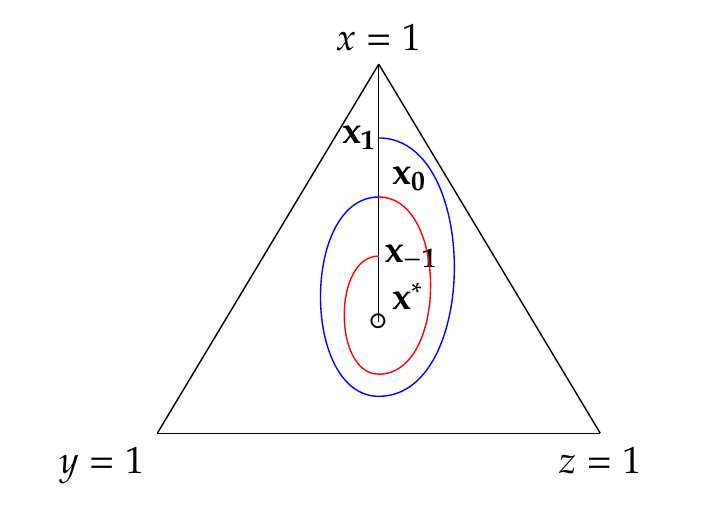}
\end{center}
\caption{The solution trajectory from $\mathbf{x_0}$ (blue) and the time-reversed solution trajectory from $\mathbf{x_0}$ (red).}
\end{figure}

The final step of the proof is to show that every solution trajectory is a closed orbit. Since all trajectories circle around $\mathbf{x^*}$, it suffices to consider a solution originating at some point $\mathbf{x_0}$ on the x-bisector between the rest point $\mathbf{x^*}$ and the state $x=1$. Suppose that once this solution completes a loop around the rest point, it hits the x-bisector again at some point $\mathbf{x_1}$, whereas if one were to reverse the flow it would hit the bisector at a point $\mathbf{x_{-1}}$. The ``self-negating'' property of $C_2$ implies that the mirror image of the phase portrait of the negative of $C_2$ is the phase portrait of $C_2$ (the x-bisector being the axis of symmetry). In particular, the mirror image of the segment of the solution trajectory between $\mathbf{x_{-1}}$ and $\mathbf{x_0}$ has to be the segment between $\mathbf{x_0}$ and $\mathbf{x_1}$. Therefore $\mathbf{x_{-1}}=\mathbf{x_1}$, but that is only possible if $\mathbf{x_{-1}}=\mathbf{x_0}=\mathbf{x_1}$, as otherwise $\mathbf{x_{-1}}$ and $\mathbf{x_1}$ must be on the opposite sides of $\mathbf{x_0}$. Therefore all solution trajectories must form closed orbits around $\mathbf{x^*}$. 
\end{proof}

The Proposition \ref{cgame} suggests that if a three-strategy game has the ``self-negating'' property, the solution trajectories and the time-reversed solution trajectories ``meet'', which implies that the trajectories form closed orbits around an interior steady state. This observation provides a link between self-negating games under the IBR dynamics and zero-sum games under the replicator dynamics. In a zero-sum game with an interior equilibrium, every interior solution trajectory is confined to a level set of a Kullback-Leibler divergence function (see Sec. 9.1.1 of \cite{Sand}). This means the rest point is Lyapunov stable but not asymptotically stable, and that other interior solution trajectories do not converge. The latter need not be the case under the IBR dynamics: game $W$ in Example \ref{Zeeman} is a self-negating game, in which some interior solution trajectories form closed orbits, while others converge to a pure rest point.

In games with more than three strategies one might not get as much mileage out of self-negation. With only three strategies, the self-negating property implies an axisymmetric phase portrait, since there must be a pair of strategies that is relabeled when the order of payoffs is reversed. With four or more strategies it is possible that more than one pair of strategies are relabeled, and it is not entirely clear what this possibility implies.  

\section{Other examples}
In this section we provide two more examples that relate the IBR and the replicator dynamics. First, we show that the game from the example 1 of Zeeman (1980) admits two interior rest points under the IBR dynamics. Such behavior is impossible under the replicator dynamics, under which in non-degenerate games there can be at most one interior rest point (Theorem 3 in \cite{Zee80}). Second, we construct a self-negating game in which the interior is split into two regions, one containing closed orbits around the interior rest point, and the other being a basin of attraction for a pure rest point. Up to the position of the rest point, both dynamics are equivalent.
\begin{example} \label{Zeeman}
Consider the games $Z$ and $W$:  
\begin{center}
\begin{tabular}{ccc}
$Z=\left(
\begin{tabular}{ccc}
0 & 6 & -4 \\
-3 & 0 & 5 \\
-1 & 3 & 0
\end{tabular}\right)$ &&
$W=\left(
\begin{tabular}{ccc}
0 & 4 & 3 \\
1 & 3 & 5 \\
3 & 2 & 6
\end{tabular}\right)$
\end{tabular}
\end{center}
Game $Z$ is the game from the Example 1 in \cite{Zee80}. The IBR dynamics in it 
\begin{align*} 
\dot{x} &= x(1-x)\left(x^2+y-z\right)-2x^2yz\\ 
\dot{y} &= -y(1-y)\left(x+y^2-z\right)+2x^2yz\\
\dot{z} & = z(x-y)\left(z-x^2-y^2\right)
\end{align*}
yields two distinct interior rest points $\mathbf{z_1}=(\frac13,\frac13,\frac13)$ and $\mathbf{z_2}\approx(0.575, 0.088, 0.338)$. At $\mathbf{z_1}$ the relevant eigenvalues are $\tfrac{1}{27}\left(-1\pm2i \sqrt5\right)$, so this rest point is stable, whereas at $\mathbf{z_2}$ the eigenvalues are 0.41 and -0.041, so it is unstable.

\begin{center}
\begin{figure}[h]
\centering
\begin{tabular}{ccc}
$\includegraphics[scale=0.5]{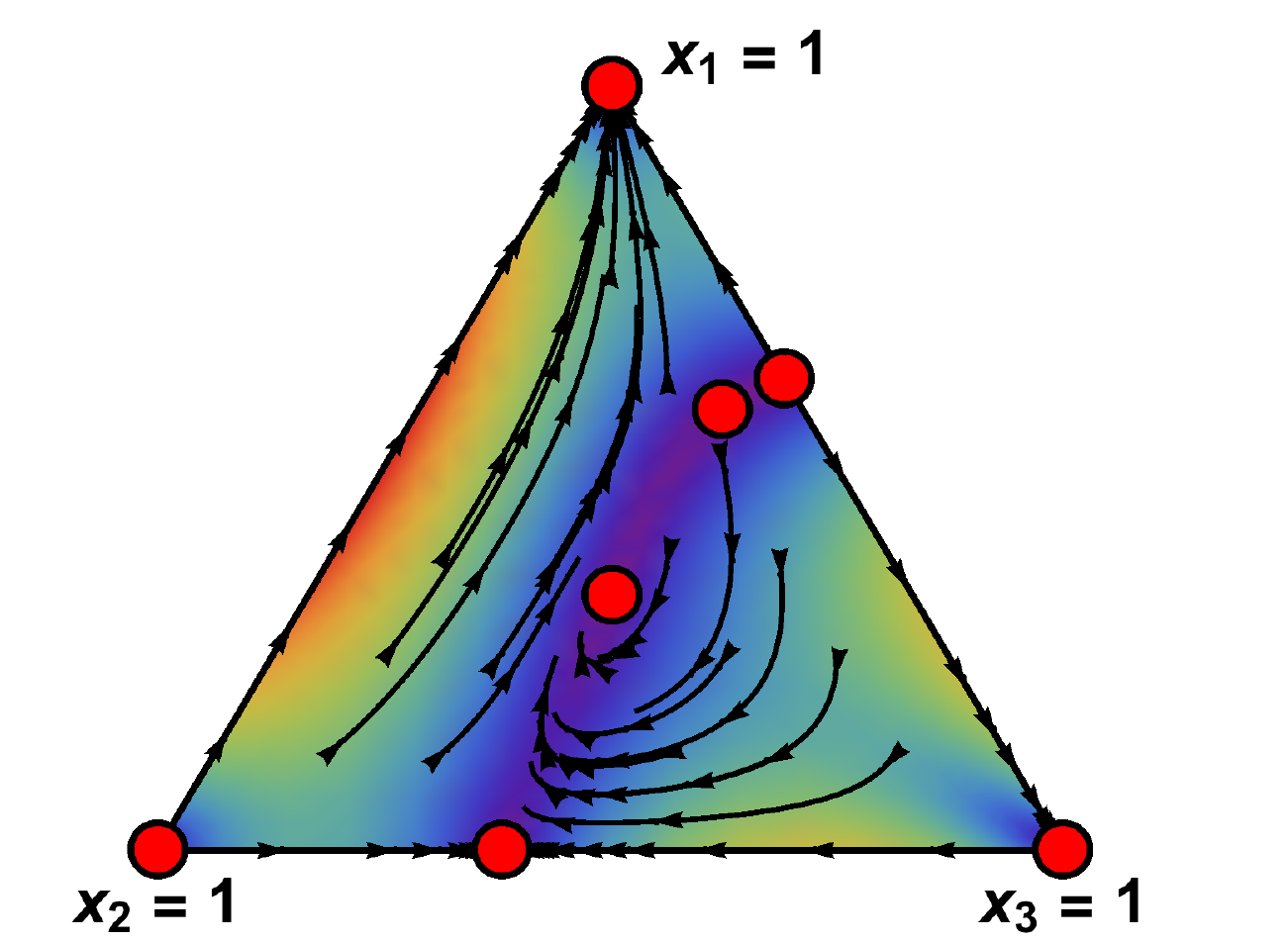}$ && $\includegraphics[scale=0.5]{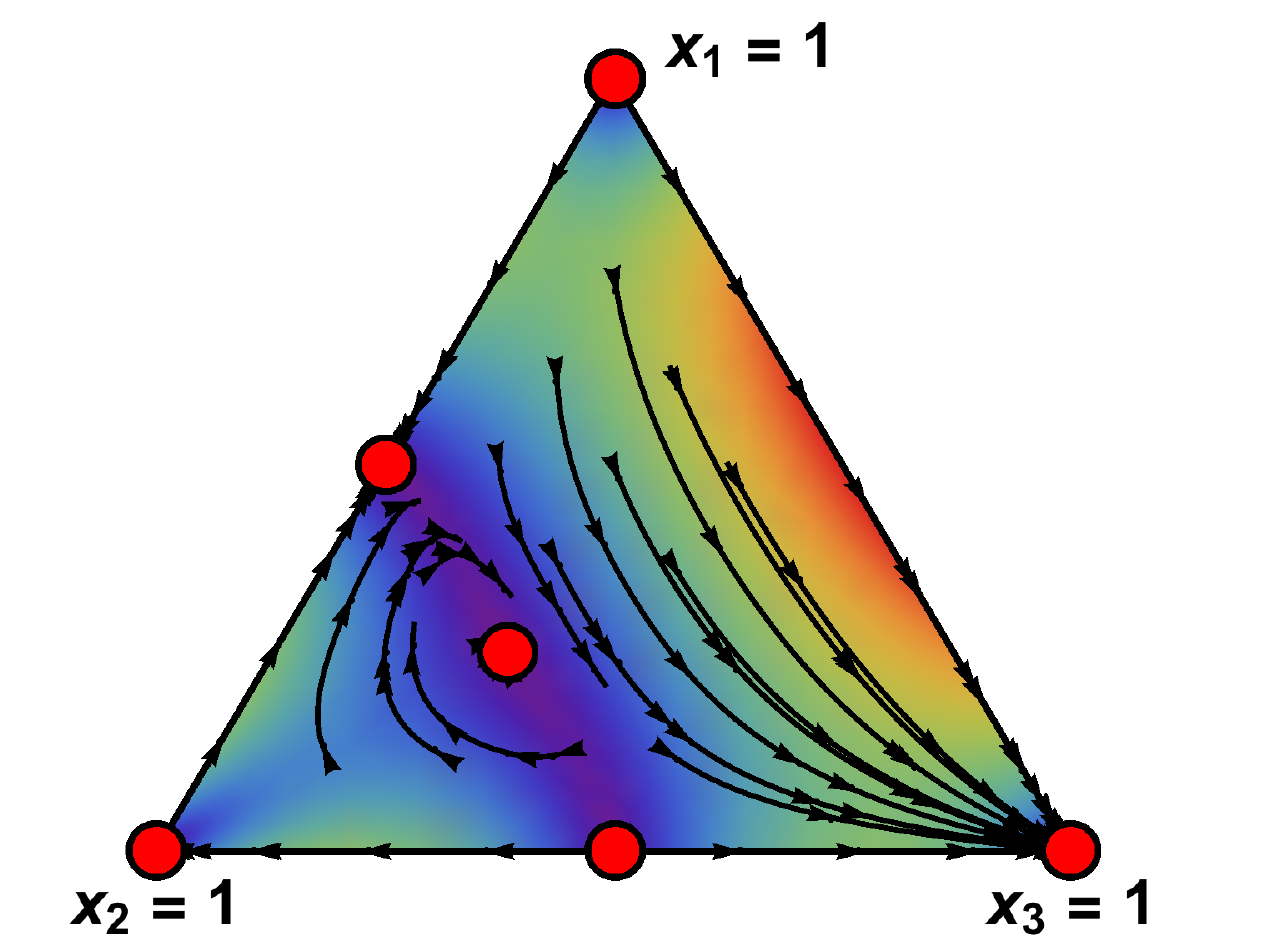}$ 
\end{tabular}
\caption{Some solution trajectories in games $Z$ (left) and $W$ (right).}
\end{figure}
\end{center}
Game $W$ has the self-negating property, so there are closed orbits around the interior rest point, but a part of the interior of the simplex is the basin of attraction of the state $x_3=1$. In the time-reversed game $x_1=1$ becomes the attractor, while the interior rest point preserves its region with the closed orbits.
\end{example}
\section{Conclusion}
This paper investigated the properties of an imitative rule that ignores any cardinal information about the game's payoffs. Agents switch to strategies which they perceive as better based on the comparison of their realized payoffs to that of a random member of the population. Since this behavioral rule bears a similarity to the pairwise proportional imitation of \cite{Sch98}, the resulting ordinal imitative dynamics begs comparison with the replicator dynamics arising from the PPI.               

We demonstrate that while the IBR dynamics does not possess the payoff monotonicity and Nash stationarity properties of the replicator dynamics in general, the two dynamics are topologically equivalent in two-strategy games. We also conjecture that they generate the same types of behavior in Rock-Paper-Scissors games. In other cases, the IBR dynamics can generate behavior that is impossible under the replicator dynamics. 

Better understanding the relationship between the two dynamics and investigating the self-negating property in games with more than three strategies would be the two most important directions for future research. 
\end{singlespace}
\end{document}